\documentclass[10pt]{article}

\makeatletter
\let\@fnsymbol\@arabic
\makeatother

\usepackage{graphicx}
\usepackage{geometry}
\usepackage{amsmath}
\usepackage{comment}
\usepackage{tabto}
\usepackage{xcolor}
\usepackage{float}
\usepackage[misc]{ifsym}
\usepackage{amssymb}
\usepackage{authblk}
\usepackage{amsthm}
\newtheorem{lemma}{Lemma}
\newtheorem{definition}{Definition}

\RequirePackage{fix-cm}
\newtheorem{theorem}{Theorem}
\DeclareMathSizes{12}{12}{10}{9}

\let\oldtheorem\theorem
\renewcommand{\theorem}{\vspace{2ex plus 0.5ex minus 0.2ex}\oldtheorem}

\let\oldlemma\lemma
\renewcommand{\lemma}{\vspace{2ex plus 0.5ex minus 0.2ex}\oldlemma}

\title{Impact of Artificial Intelligence on the Environment through Technical Change: A Free Dynamic Equilibrium Approach}

\author{Pham Van Khanh}
\affil{Vietnam Academy of Science and Technology, Hanoi, Vietnam.\hspace{10mm} \Letter \quad van\_khanh1178@yahoo.com}
\author[2]{Minh Le}
\affil{Georgia Institute of Technology, Atlanta, Georgia, United States. \hspace{20mm} \Letter \quad dd350@gatech.edu}

\date{}

\begin{document}
\maketitle
	
\abstract{In today's world, humanity faces unprecedented environmental challenges, such as climate change and natural disasters. The emergence of artificial intelligence (AI) has opened new doors in our efforts to address our planet's pressing problems. However, many doubt the actual extent of impact that AI can have on the environment. Particularly, AI's role in dirty production is a drawback that is largely absent from the literature. To investigate the impact of AI on the environment, we establish mathematical models to demonstrate the economy and the production process using outdated and advanced technologies. The secondary results are stated as lemmas, the main results are stated as theorems. From the theorems we conclude that AI may not on its own prevent an environmental disaster, partially due to its concurrent contribution to dirty production. With temporary government intervention, however, AI is able to avert an environmental disaster.}
 
 \vspace{3mm} \noindent
 	
 \vspace{3mm} \noindent
 \textbf{Funding:} The authors did not receive support from any organization for the submitted work.
	
\vspace{3mm} \noindent
\textbf{Conflict of Interest:} The authors have no relevant financial or non-financial interests to disclose.
	
\vspace{3mm} \noindent
\textbf{Ethical Approval:} This article does not contain any studies with human participants or animals performed by any of the authors.
	
	
\vspace{20cm}
	
\setlength{\parskip}{2ex plus 0.5ex minus 0.2ex}
\section{Introduction}
The increase in the average global temperature on Earth is one of the main factors causing climate change (Stern \& Robert, 2014). In 2020, average annual temperatures increased by 1.25 degree Celsius over the past century and by about 0.54 degrees Celsius over the past 30 years alone (NASA, 2023).
 
Climate change causes significant damages to the economy and the well-being of mankind. Regarding the economy, it is shown via data analysis that a wider temperature range reduces capital investment and thus hinders economic growth (Lu et al., 2019). Climate change disproportionately affects poorer countries (Tol, 2018), thus worsening global inequality and making it even more difficult for countries to escape poverty. It also is a main driver of some extreme weather events such as cyclones and storms (Clarke, 2022), which can destroy physical capital such as agricultural land, factories and buildings (Batten, 2018).

With respect to our physical well-being, major storms (whose increasing frequency are driven by climate change) can cause major injuries as people are hit by flying debris (Alderman, Turner \& Tong, 2012). Moreover, extreme heat increases the ozone concentration on the ground, which can result in breathing problems and lung diseases (Crimmins, 2016).

Moreover, chronic weather conditions such as droughts worsens food security, resulting in increasing stress and anxiety (Manning \& Clayton, 2018). High temperatures are associated with higher suicide rate (Lawrance, 2021). Extreme events such as hurricanes can cause long-lasting trauma and depression due to injury of oneself or loved ones, property destruction and adverse living conditions (Galea, 2007).

With such extreme consequences, climate change must be addressed, and to achieve this goal, we need to evaluate the causes of climate change. This is not an easy task and while research is still in progress, a majority of experts agree that anthropogenic climate change has a significant impact on our ecosystems (see e.g. Rosenzweig et al., 2008; Kaufmann et al., 2011; Hansen \& Stone, 2016)

The fact that climate change is highly likely anthropogenic presents a unique challenge: that many attempts at slowing down climate change go against economic growth, the main momentum of the vast majorities of economies. In fact, some are so concerned they believe it is unlikely to be solved, and instead propose \textit{degrowth}: a reduction in the world's GDP. A significant portion of the literature (e.g. Kallis, 2011; Sekulova et al., 2013; Morgan, 2020) where authors argue that attempts to phase out dirty energy and production cannot succeed while the world is still chasing GDP growth.
 
Within this discussion, artificial intelligence (AI) is receiving more and more attention as a potential solution to tip the scale towards sustainable growth, for many AI applications have shown significant benefits for the environment without impeding economic growth. AI is used for climate modeling and prediction (Chantry et al., 2021). AI-driven systems optimize renewable energy harvesting, improving efficiency and sustainability in the energy sector (Patel et al., 2022). In smart buildings, AI manages energy consumption, reducing costs and emissions (Farzaneh et al., 2021). In agriculture, smart agriculture leverages sensors and data analytics to maximize crop yields while conserving land and water resources (Hemming et al. 2020), and early detection of crop diseases and pests allows a reduction in toxic pesticide usage (Selvaraj et al., 2019).

AI also plays a large role in enhancing economic growth. We show in section 2 that with an abundance of different AI and machine learning models, AI provides a wide-ranging and effective tool set to solve problems in research, thus increasing the likelihood of success and the real-world impact of research. This translates to improvements in productivity and efficiency in industries. Developments in AI, therefore, are essentially supply-side policies that enhance economic growth. Significantly, this also signals that environmental protection and economic growth are not mutually exclusive and can go in tandem together. Thus, microeconomics equilibrium that manages to heavily involve AI may actually improve environmental quality, and policies to correct externalities may not be necessary.

However, critics have raised concerns about the actual ability of AI to drive economic growth and environmental protection. Regarding economic growth, although the literature generally agrees that AI has a positive impact on GDP, they disagree on the extent due to limitations of AI. We will argue in section 2 that these limitations are likely to be addressed in the future, and thus AI likely will have a large positive impact on the economy. Regarding environmental protection, its impact in addressing climate change is difficult to fully assess also due to limitations of AI. A major limitation that is popular in the literature is the large carbon footprint of training AI model (Tamburrini, 2022).

Our paper contributes to this literature by investigating another largely unexplored limitation: The extent to which AI accelerates dirty production, and its consequences on the environment. Amidst talks about how AI drives clean innovation, we are cautious about the optimism surrounding them because we need to recognise and address the fact that AI can be applied to the dirty sector too.

Research on this is very scarce; we can only locate Andres (2022) who claims that AI innovation benefits clean production three times more than dirty production. Our paper will be built on that result and explore whether this allows the clean sector to overtake the dirty sector.

In Acemoglu et al. (2012), the authors prove that an environmental disaster cannot be avoided under economic equilibrium. We will modify the mathematical model in this paper to add the impact of AI in production. We will assume a highly optimistic scenario for AI development, that AI progress grows at an exponential rate.

Our paper will show that despite AI seemingly large potential for the environment, it is uncertain that AI alone can prevent an environmental disaster in economic equilibrium. We further argue qualitatively that it is unlikely for AI to achieve such a goal. Only when AI is combined with temporary government intervention, can it prevent an environmental disaster. We will show in our discussion of the result that the contribution of AI to the dirty sector is a major reason for AI's inability to achieve this goal.
 
Our paper is structured as follows: part 1 is for the introduction, part 2 is to access the impact of AI on the economy, part 3 is for proposal and model building, part 4 includes the results of the research, part 5 is for discussion of the results, part 6 is about the computer simulation of the model, and finally part 7 is for the conclusion.

\section{Impact of AI on Research and Production}

This paper uses a mathematical model in which scientific research improves the quality of machinery, which increases productivity. In this section, we argue that AI can enhance productivity, in order to integrate its impact into our model. We also address criticisms of AI and attempt to quantify the impact of AI based on predictions from other researches.

\subsection{How AI Increases Productivity}

We explains how AI either supplements science research or is the product of research itself, through which productivity is improved. Two review articles by Wang et al. (2023) and Xu et al. (2021) discuss the role of AI through different lenses. Wang et al. focuses on the ways and channels through which AI aids research, while Xu et al. explores the fields of research in which AI has major positive impacts.

In Wang et al., the authors explains that AI can assist research at many junctures, which demonstrates how AI supplements science research. Some of those junctures are:

\begin{itemize}
    \item \textit{Data generation:} AI itself can be used to generate additional data from existing dataset to train the main AI model. Klenn \& Bergmann (2019) uses this approach to develop an AI model that detects failures in a factory, where raw training data is scarce and inaccessible. This allows for timely maintenance, which reduces downtime and increases productivity.
    
    \item \textit{Self-supervised learning:} When data labelling is not available and is too costly to be done, self-supervised learning can be used to learn the dataset without labels and generate its own features. In Magalhães et al. (2023), self-supervised learning is used to ``classifies anomalies in an industrial space", making this task more efficient.

    \item \textit{Hypothesis generation:} A hypothesis is an educated guess of the answer to a problem, which is to be tested by scientific research. However, it can be difficult to obtain a good guess in the first place for complicated, obscure problems. AI can help generating hypotheses by exploring the data given. Liu et al. (2023) uses a deep learning model to hypothesise possible antibiotics that are efficient in killing a certain bacteria.

    \item \textit{Simulation:} In cases where physical experiments are too costly or impractical, computer simulation, many of which are assisted by AI, can be used. Bruzzone \& Orsoni (2003) develop an AI-based model that simulate the logistic of a supply chain to evaluate proposed solution to increase efficiency.
\end{itemize}

On the other hand, Xu et al. surveys the application of AI in different fields, showing how researches produce AI application that enhances production. They include:

\begin{itemize}
    \item \textit{Information Technology:} AutoML is a tool that automate the machine learning development workflow, allowing non-experts in machine learning to still apply it to their problems. For example, de Souza et al. (2021) develop Spectral AutoML specifically to increase the efficiency of soft sensor development.

    \item \textit{Materials Science}: AI models can be used to predict properties of new materials, aiding material discovery and design as scientists can more quickly test properties of their material designs and tweak them as needed. Yazdani-Asrami et al. (2022) outlines how this is done in applied superconductivity, increasing productivity in this task.

    \item \textit{Geoscience:} Managing the water resource of an area is challenging, especially in the context of climate change bringing in additional uncertainty. AI models, such as one developed by Xiang et al. (2021), can perform this task more accurately.
\end{itemize}

These two review articles, combined with many industrial-related examples provided above, have shown the two channels through which AI contributes to increasing productivity. Further more, they also demonstrate both the breadth and depth of the impact of AI in research to aid productivity, as AI is used for multiple tasks in research, for research in many different fields. Thus, AI has an immense potential in propelling our scientific progress for production.

In addition, Andres (2022) shows that the clean production sector benefits three times more than the dirty sector from AI innovation. This observation still holds (albeit with a smaller multiplier) after fixing firm effects, showing that this property result from the inherent difference between clean and dirty production, instead from the difference in investment in research and development between firms. Therefore, this trend is highly likely to still hold far into the future.

\subsection{Evaluation of Criticisms of AI}
Despite its large potential, there are still many problems in AI and its application in research and production. Because of that, many scientists have rightfully raised concerns about AI, or cast doubt on the extent of impact AI can have. To arrive at a more nuanced, supported observation of the potential of AI, we evaluate some of its common criticisms.

Our central claim is that while the criticisms are largely valid, they are likely to be addressed and mitigated in the future. To argue for this, we will show two components:

\begin{itemize}
    \item There are existing researches that tackle those criticisms, and
    \item There will likely be significant progress in these research areas in the future.
\end{itemize}

We will thus conclude that AI will likely still have a significant positive impact on research and production.

\subsubsection{AI can be inaccurate}

No AI system is totally accurate. The best AI models have accuracy of around 90\%, meaning it is wrong in 1 out of 10 cases. This heavily limits the impact of AI on research and industry. Kwon, Raman, \& Moreno (2023) show that AI inaccuracy caused by bad training input makes manpower allocation inefficient, demonstrating its impact in an industrial setting. Might this be an eternal problem for AI, one that can't be overcome?

We believe the answer is no, because in practice, AI can be used with human to improve efficiency. Wilson \& Daugherty (2018) even predict that a future where AI collaborate with human is more likely than AI replacing us. To that extent, researches have put forward many ways to improve human-AI collaboration, which include increasing AI confidence (Zhang et al., 2020) and transparency (Vössing et al., 2022), providing descriptions of AI behaviour (Cabrera et al., 2023), and enhancing user decision control (Westphal et al., 2023).

Moreover, the large number of researches already existed shows that human-AI collaboration has the potential to be a highly effective solution for our issue. Recognising this, more researchers will focus on this research area, hoping that their works contribute to the literature that will be applied and have significant positive impact on the real world. Hence, we believe that there will likely be significant progress in human-AI research in the future.

\subsubsection{AI based on supervised learning only learn and cannot improve form the past}

Supervised learning is a subset of machine learning, which is in turn a subset of AI. Supervised learning refers to deducing patterns from labelled data to perform a task. If the task is classification (e.g. is this image of a cat or a dog?), the label will be one of the category. If the task is regression (e.g. how much should this house cost?), the label will be a numerical value.

In Hagendorff \& Wezel (2020), the author focus on supervised learning as a representative of AI because it ``is the method used in the vast majority of artificially intelligent applications." They criticise supervised learning models, arguing that because they are trained on past, existing data, they are only able to replicate the past instead of coming up with novel and creative ideas. They rebut counterexamples of arts by AI by claiming that they are only ``deformed re-configurations of existing art works." This ``status quo" preference is problematic if it reinforces the issues we are trying to change. For example, in an attempt to use AI to eliminate human bias in some justice systems, the AI actually reflects those existing biases (Malek, 2022).

We largely agree with the authors' usage of supervised learning as a representative of AI, as well as their criticism of supervised learning. However, we also believe this approach makes the authors' argument less applicable for the future where self-supervised learning is becoming a new frontier of AI research.

Self-supervised learning is also a form of machine learning that is used for similar tasks with supervised learning, but is one that uses unlabeled data for model training. This allows the model to deduce its own connections and come to its own conclusions which can be novel from human, reflecting an attempt to move away and improve from the status quo. For example, Sirotkin (2022) shows that social bias can be reduced with a careful choice of self-supervised model.

Moreover, we believe that self-supervised learning is highly likely to see significant progress in the future and might even replace supervised learning in many applications. Bergmann (2023) notes that acquiring labelled datasets for training can be difficult since there may be no existing dataset, and data labelling is a long, expensive and manually tedious process. Hence, self-supervised learning is an excellent alternative when acquiring labelled datasets is challenging. With machine learning being employed in more diverse, more challenging tasks with fewer labelled datasets available, self-supervised learning will be the new frontier in AI development, and research will surely focus more on this area in the future.

\subsubsection{AI faces difficulties in harder tasks}

In the discussion of AI contribution to the macroeconomy, Acemoglu (2024) differentiates tasks for AI to solve into easy and hard tasks. Easy tasks, he wrote, ``\textit{are defined by two characteristics}:

\begin{itemize}
    \item \textit{there is a simple (low-dimensional) mapping between action and the (perfect) outcome measure, and}
    \item \textit{there is a reliable, observable outcome metric.}"
\end{itemize}

Hard tasks are the opposite, meaning they either are highly complex, or their outcomes cannot be reliably or fully observed. Acemoglu reasonably argues that these difficulties heavily hinder the ability of AI to effectively solve the tasks, and hence questions the extent of contribution of AI to productivity and growth.

It however must be noted that Acemoglu chooses to focus on the short term future of the next 10 years only. In the longer term, which is the scope of our article, we claim that there are already solutions to address both the complexity and observability issues, with promising significant progress in the future.

Regarding complexity of tasks, a review paper by Joksimovic (2023) finds that human-AI collaboration has been a main solution explored in a number of papers, which has seen an increase in quality over time. The argument that human-AI collaboration will likely see significant progress in the future has been laid out in section 2.1.

Regarding reliable observability of results, reinforcement learning and its modifications have been used to solve problems under partially observable environments. In simple terms, a reinforcement learning models involves an actor and an environment, where in each turn, the actor choose an action to perform, insodoing interacts and changes the environment, and is granted a reward by the environment. The actor's goal is usually to maximise the sum of rewards over time.

Much progress has been made to improve reinforcement learning in partially observable environments. Spaan (2006) proposes approximate planning in his thesis. Choi and Kim (2011) explores recovering the reward function of the environment. Muškardin et al. (2023) goes further and suggests learning the entire environment model. Srinivasan \& Lanctot et al. (2018) applies the actor-critic model for this task, where the actor employs a `critic' intelligent model (to be trained) to evaluate the actor's actions.

Similar to human-AI collaboration, the vast number of papers and approaches in reinforcement learning for partially observable environments shows that it is a highly promising solution to the issue at hand, which is recognised by researchers to focus their research on, fuelling even more progress in the field.

\subsection{Quantifying the Impact of AI}

To quantify the impact of AI on production, we will assume that the percentage increase in production of the final good is the same as the percentage increase in GDP. We then review the literature on predictions of impact of AI on GDP growth, and find that they vary significantly. Acemoglu (2024) noted that quantifying the impact of AI on the economy is ``\textit{extremely difficult and will have to be based on a number of speculative assumptions}," explaining this situation.

To name a few, Acemoglu (2024) takes a relatively cautious stand, only predicting a contribution of 0.9\% over the next 10 years. Others are more optimistic: Goldman Sachs (2023) headline reads ``\textit{Generative AI could raise global GDP by 7\%}", over a 10-year period.

Some envision significantly larger growth. McKinsey \& Company (2023) predicts a \$2.6-\$4.4 trillion annual contribution to the global economy. Noting that the 2023 global GDP is \$105.44 trillion (World Bank, 2024), this contribution corresponds to an average of 3.32\% annual growth, or an estimate of 33.2\% growth over 10 years.

Korinek and Suh (2024) even claim a 100\% GDP growth over 10 years with the help of Artificial General Intelligence. Since the global GDP increased by 35.6\% from 2013 to 2023 (World Bank, 2024), this implies that AI is responsible for 64.4\% of GDP growth in a decade.

In light of such varied predictions, we will use all four of the above values in checking whether our result on avoiding an environmental disaster holds under them, and in conducting computer simulations.
	
\section{The Model}
The framework uses an economy where time is discrete and goes from $0$ to infinity. There is a unique final good produced competitively using inputs produced by either clean or dirty technology. Scientists conduct research to improve the quality of machines to increase production. Dirty production degrades the environment, which after a certain point becomes an environmental disaster.

\subsection{The Economy}
	
At time $t$, the economy produces $Y_t$ units of the unique final good using clean and dirty inputs $Y_c$ and $Y_d$ respectively. According to the aggregate production function:
\begin{equation} \label{Yt}
    Y_t=\left(Y_{ct}^\frac{\sigma-1}{\sigma}+Y_{dt}^\frac{\sigma-1}{\sigma}\right)^\frac{\sigma}{\sigma-1}
\end{equation}

where $\sigma>0$ is the elasticity of substitution between the clean and dirty sectors. We ignore the distribution parameter for simplicity. Throughout the paper, we assume that $\sigma>1$, meaning that clean and dirty inputs are gross substitutes.

The two inputs $Y_c$ and $Y_d$ are produced using labor and a continuum of machines that are sector-specific, meaning they can only be used on one sector, clean or dirty. They are calculated by:
\begin{equation} \label{Yjt}
    Y_{jt}=L_{jt}^{1-\alpha}\int_0^1A_{jit}^{1-\alpha}x_{jit}^\alpha \ di
\end{equation}

where $\alpha \in (0,1)$, $j\in \{c,d\}$, $A_{jit}$ is the quality of machine of type $i$ used in sector $j$ at time $t$, and $x_{jit}$ is the quantity of this machine.

Market clearing for labor means labor demand is less than labour supply. Normalising the latter to 1, we get:
\begin{equation} \label{Ljt_2}
    L_{ct}+L_{dt}\leq 1
\end{equation}

Market clearing for the final good means:
\begin{equation} \label{nt}
    C_t=Y_t-\psi\left(\int_0^1x_{cit}\ di+\int_0^1x_{dit}\ di\right)
\end{equation}

where at time $t$, $C_t$ is the consumption of the final good, and $\psi$ is the number of final good units needed to produce a machine.

Let $p_{ct}$ and $p_{dt}$ be the price of clean and dirty inputs respectively. Since clean and dirty inputs are used competitively to produce the final good, we have:
\begin{equation} \label{pjt}
    \frac{p_{ct}}{p_{dt}}=\left(\frac{Y_{ct}}{Y_{dt}}\right)^{-\frac{1}{\sigma}}
\end{equation}

In addition, by normalising the price of the final good to one, we get:
\begin{equation} \label{pjt_3}
    \left(p_{ct}^{1-\sigma}+p_{dt}^{1-\sigma}\right)^\frac{1}{1-\sigma}=1
\end{equation}

\subsection{The Innovation Process}

Scientists conduct research with the help of AI to improve the quality of machines. Scientists can choose whether to research in the clean or dirty sector based on higher expected profit. We assume that the quality of AI for each sector increases at a fix rate per decade. This means that AI quality in sector $j$ at time $t$, $I_{jt}$, grows exponentially.

Since AI contributes to the clean sector three times more than to the dirty sector, we model $I_{jt}$ as:
\begin{equation} \label{Ict}
    I_{ct}=e^{3kt}, \, I_{dt}=e^{kt}
\end{equation}

where $3k$ and $k$ denotes the rate of development of AI in the clean and dirty sector respectively. Note that there is no additional constant term accompanying the exponent because we assume that at $t=0$, AI has negligible impact on production, meaning $I_{jt}=1$, yielding a constant term of 1.

The probability that a scientist is successful in innovation in sector $j$ is $\eta_j$. If successful, the research improves the quality of a machine in the dirty sector by a factor of $1+\gamma$ (where $\gamma>0$), from $A_{dit}$ to $(1+\gamma)A_{dit}$.

In the clean sector however, the effectiveness of successful research is greater by a factor of $I_t$ due to the impact of AI. Quality of machines thus increases by a factor of $1+\gamma I_t$, from $A_{cit}$ to $(1+\gamma e^{kt})A_{cit}$

Market clearing for scientists means:
\begin{equation} \label{sjt}
    s_{ct}+s_{dt}\leq 1
\end{equation}

We also define:
\begin{equation} \label{Ajt}
    A_{jt}=\int_0^1A_{jit}\ di
\end{equation}

The process of innovation is then described as:
\begin{equation} \label{Ajt_2}
    A_{jt}=(1+\gamma\eta_jI_{jt}s_{jt})A_{jt-1}
\end{equation}

\subsection{The Environment}

Let $S_t$ is the quality of the environment. We let $S\in [0,\bar{S}]$, where $\bar{S}$ is the quality of the environment without any human pollution. We also assume that it is the initial level of environmental quality, meaning $S_0=\bar{S}$.

The quality of the environment evolves according to:
\begin{equation} \label{St}
    S_{t+1}=-\xi Y_{dt}+(1+\delta)S_t
\end{equation}

when the right-hand side (RHS) is within $[0,\bar{S}]$. When the RHS is negative, $S_{t+1}=0$, and when it is larger than $\bar{S}$, $S_{t+1}=
\bar{S}$. Here, $\xi$ denotes the extent to which production in the dirty sector has a negative impact on the environment, and $\delta$ is the rate at which the environment regenerates itself.

For notational convenience, we also define
\begin{equation}
    \varphi=(1-\alpha)(1-\sigma)
\end{equation}

Finally, we define the following term:

\begin{definition}
    An environmental disaster refers to when $S_t=0$ for some finite $t$.
\end{definition}

\section{Free Dynamic Equilibrium}

\subsection{Preliminary Setups}

We first define the free dynamic equilibrium:

\begin{definition}
    A free dynamic equilibrium is given by sequences of wages of workers ($w_t$), prices of inputs ($p_{jt}$), prices of machines ($p_{jit}$),  demand for machines ($x_{jit}$), demand for labour ($L_{jt}$), scientist allocations ($s_{ct}$, $s_{dt}$) and quality of the environment ($S_t$) such that, in all time $t$:
    \begin{enumerate}
        \item ($p_{jit}$, $x_{jit}$) maximises profits of the producer of machine $i$ in sector $j$.
        \item $L_{jt}$ maximises profits of producers of input $j$.
        \item $Y_{jt}$ maximises the profits of final good producers.
        \item ($s_{ct}$, $s_{dt}$) maximises the expected profit of a scientist at date $t$.
        \item The wage $w_t$ and prices $p_{jt}$ clear the labour and input markets, respectively.
        \item $S_t$ changes are given by (\ref{St}).
    \end{enumerate}
\end{definition}

\subsection{Production in Equilibrium}

This section aims to express the ratio expected profits of scientists in the clean and dirty sector in terms of constant parameters and variables relevant to directed technical change, such as $\gamma$, $\eta_j$, $I_t$, $s_{jt}$ and $A_{jt}$.

    Firstly, we calculate $x_{jit}$, the demand for machines type $i$ in sector $j$ at time $t$.

\begin{lemma}
    \begin{equation} \label{xjit_2}
        x_{jit}=\left(\frac{\alpha^2}{\psi}\right)^\frac{1}{1-\alpha}p_{jt}^\frac{1}{1-\alpha}L_{jt}A_{jt}
    \end{equation}
\end{lemma}

\begin{proof}
    Per Definition 2, condition 2, the profit-maximisation problem of the producer of input $j$ at time $t$ is:
    \begin{equation} \label{profit}
        \max_{L_{jt}}\left\{p_{jt} L_{jt}^{1-\alpha} \int_0^1 A_{jit}^{1-\alpha} x_{jit}^{\alpha} \; di - w_t L_{jt} - \int_0^1 p_{jit} x_{jit} \; di\right\}
    \end{equation}
    
    where $x_{jit}$ is chosen by the producer of machine $i$ in sector $j$ at time $t$ to maximise its profit (per Definition 2, condition 1). Note that the terms correspond to revenue of input produced, cost of labour and cost of machines respectively.
    
    To get the inverse curve of the input producer, we partially differentiate with respect to (w.r.t) $x_{jit}$:
    \begin{equation} \label{xjit}
        0 = p_{jt}L_{jt}^{1-\alpha}A_{jit}^{1-\alpha}\alpha x_{jit}^{\alpha-1}-p_{jit} \Rightarrow x_{jit}=\left(\frac{p_{jit}}{\alpha p_{jt}}\right)^\frac{1}{\alpha-1}L_{jt}A_{jit}
    \end{equation}
    
    This is an iso-elastic inverse demand curve, and therefore its price elasticity of (PED) is the power of the price of input $p_{jit}$, which is $1/(\alpha-1)$.
    
    According to the mark-up rule, the profit maximising price for the machine producer is a constant markup over marginal cost by a factor of:
    $$\frac{1}{1+1/PED}=\frac{1}{1+\alpha-1}=\frac{1}{\alpha}$$
    
    With the marginal cost being $\psi$, we have the profit maximising price $p_{jit}=\psi/\alpha$. Plugging this into (\ref{xjit}) gets us:
    \begin{equation*}
        x_{jit}=\left(\frac{p_{jit}}{\alpha p_{jt}}\right)^\frac{1}{\alpha-1}L_{jt}A_{jit}=\left(\frac{\alpha^2}{\psi}\right)^\frac{1}{1-\alpha}p_{jt}^\frac{1}{1-\alpha}L_{jt}A_{jt}
    \end{equation*}
\end{proof}

    Next, we calculate the expected profit of a scientist in sector $j$ at time $t$.
    
\begin{lemma}
    Let $\Pi_{jt}$ be the expected profit of a scientist in sector $j$ at time $t$. Then:
    \begin{equation} \label{pijt_3}
        \Pi_{jt}=\eta_j(1+\gamma I_{jt})(1-\alpha)\alpha^{\frac{1+\alpha}{1-\alpha}}\psi^{\frac{-\alpha}{1-\alpha}}p_{jt}^\frac{1}{1-\alpha}L_{jt} A_{jt-1}
    \end{equation}
\end{lemma}

\begin{proof}
    Since the machine producer sells $x_{jit}$ machines at the price of $p_{jit}$, and each of which costs $\psi$, its equilibrium profit is:
    \begin{multline}
        \pi_{jit}=(p_{jit}-\psi)x_{jit}=\left(\frac{\psi}{\alpha}-\psi\right)\left(\frac{\alpha^2}{\psi}\right)^\frac{1}{1-\alpha}p_{jt}^\frac{1}{1-\alpha}L_{jt}A_{jit}\\
        =(1-\alpha)\alpha^{\frac{2}{1-\alpha}-1}\psi^{1-\frac{1}{1-\alpha}}p_{jt}^\frac{1}{1-\alpha}L_{jt}A_{jit}=(1-\alpha)\alpha^{\frac{1+\alpha}{1-\alpha}}\psi^{\frac{-\alpha}{1-\alpha}}p_{jt}^\frac{1}{1-\alpha}L_{jt}A_{jit}
    \end{multline}
    
    To get the expected profit of a scientist doing research in sector $j$ at time $t$, we take into account the probability of success $\eta_j$ and the increase in innovation $(1+\gamma I_{jt})$ if successful. Using equation (\ref{Ajt}), we get:
    \begin{multline*}
        \Pi_{jt}=\int_0^1 \pi_{jit} \; di=\eta_j(1+\gamma I_{jt})(1-\alpha)\alpha^{\frac{1+\alpha}{1-\alpha}}\psi^{\frac{-\alpha}{1-\alpha}}p_{jt}^\frac{1}{1-\alpha}L_{jt}\int_0^1 A_{jit-1}\ di\\
        =\eta_j(1+\gamma I_{jt})(1-\alpha)\alpha^{\frac{1+\alpha}{1-\alpha}}\psi^{\frac{-\alpha}{1-\alpha}}p_{jt}^\frac{1}{1-\alpha}L_{jt} A_{jt-1}
    \end{multline*}
\end{proof}

    Lemma 2 shows that the expected profit of a scientist in the clean sector is boosted by and increases with the development of AI. Thus, AI development is important for scientific progress, particularly in the clean sector.

    The next lemma calculates the price of machines in terms of productivity in the sectors.
    
\begin{lemma}
    \begin{equation} \label{pnt}
        p_{ct}=\left(\frac{A_{ct}^{-\varphi}}{A_{ct}^{-\varphi}+A_{dt}^{-\varphi}}\right)^\frac{1}{1-\sigma}
    \end{equation}
    \begin{equation} \label{pot}
        p_{dt}=\left(\frac{A_{dt}^{-\varphi}}{A_{ct}^{-\varphi}+A_{dt}^{-\varphi}}\right)^\frac{1}{1-\sigma}
    \end{equation}
\end{lemma}

\begin{proof}
    Partially differentiating (\ref{profit}) w.r.t $L_{jt}$, we get:
    \begin{multline}
        w_i=(1-\alpha)p_{jt}L_{jt}^{-\alpha}\int_0^1A_{jit}^{1-\alpha}x_{jit}^\alpha \; di=\left(\frac{\alpha^2}{\psi}\right)^\frac{1}{1-\alpha}(1-\alpha)p_{jt}L_{jt}^{-\alpha}\int_0^1A_{jit}^{1-\alpha}p_{jt}^\frac{\alpha}{1-\alpha}L_{jt}^\alpha A_{jit}^\alpha \; di\\
        =\left(\frac{\alpha^2}{\psi}\right)^\frac{1}{1-\alpha}(1-\alpha)p_{jt}^\frac{1}{1-\alpha}\int_0^1A_{jit} \; di=\left(\frac{\alpha^2}{\psi}\right)^\frac{1}{1-\alpha}(1-\alpha)p_{jt}^\frac{1}{1-\alpha}A_{jt}
    \end{multline}
    
    Since the wage $w_i$ applies to both sectors:
    \begin{equation} \label{pjt_2}
        p_{ct}^\frac{1}{1-\alpha}A_{ct}=p_{dt}^\frac{1}{1-\alpha}A_{dt}\Rightarrow \frac{p_{ct}}{p_{dt}}=\left(\frac{A_{ct}}{A_{dt}}\right)^{\alpha-1}
    \end{equation}
    \begin{equation*}
        \Rightarrow \frac{p_{ct}^{1-\sigma}}{p_{dt}^{1-\sigma}}=\left(\frac{A_{ct}}{A_{dt}}\right)^{-\varphi} \Rightarrow \frac{p_{ct}^{1-\sigma}}{p_{ct}^{1-\sigma} + p_{dt}^{1-\sigma}}=\frac{A_{ct}^{-\varphi}}{A_{ct}^{-\varphi}+A_{dt}^{-\varphi}}
    \end{equation*}
    
    Since (\ref{pjt_3}) implies that $p_{ct}^{1-\sigma} + p_{dt}^{1-\sigma}=1$, the above immediately implies (\ref{pnt}). A similar argument results in (\ref{pot})
\end{proof}

    We then calculate the distribution of labour between the two sectors.
    
\begin{lemma}
    \begin{equation} \label{Lnt}
        L_{ct}=\frac{A_{ct}^{-\varphi}}{A_{ct}^{-\varphi}+A_{dt}^{-\varphi}}
    \end{equation}
    \begin{equation} \label{Lot}
        L_{dt}=\frac{A_{dt}^{-\varphi}}{A_{ct}^{-\varphi}+A_{dt}^{-\varphi}}
    \end{equation}
\end{lemma}

\begin{proof}
    Using (\ref{pjt}) and (\ref{pjt_2})
    \begin{multline} \label{Ljt}
        \frac{L_{ct}}{L_{dt}}=\left(\frac{p_{ct}}{p_{dt}}\right)^{-\alpha/(1-\alpha)}\left(\frac{A_{ct}}{A_{dt}}\right)^{-1}\frac{Y_{ct}}{Y_{dt}}=\left(\frac{p_{ct}}{p_{dt}}\right)^{-\alpha/(1-\alpha)}\left(\frac{A_{ct}}{A_{dt}}\right)^{-1}\left(\frac{p_{ct}}{p_{dt}}\right)^{-\sigma}\\
        =\left(\frac{A_{ct}}{A_{dt}}\right)^\alpha\left(\frac{A_{ct}}{A_{dt}}\right)^{-1}\left(\frac{A_{ct}}{A_{dt}}\right)^{-\sigma(\alpha-1)}=\left(\frac{A_{ct}}{A_{dt}}\right)^{\alpha-1-\sigma\alpha+\sigma}=\left(\frac{A_{ct}}{A_{dt}}\right)^{-(1-\alpha)(1-\sigma)}=\left(\frac{A_{ct}}{A_{dt}}\right)^{-\varphi}
    \end{multline}
    
    Then, using (\ref{Ljt_2}):
    \begin{equation*}
        L_{ct}=\frac{L_{ct}}{L_{ct}+L_{dt}}=\frac{A_{ct}^{-\varphi}}{A_{ct}^{-\varphi}+A_{dt}^{-\varphi}}
    \end{equation*}
    
    giving (\ref{Lnt}). A similar argument gives (\ref{Lot}).
\end{proof}

    Noting that $-\varphi>0$, this lemma shows that labour will be distributed more towards the sector with a higher productivity. Research to increase productivity thus also has the effect of attracting labour for production.

    The following lemma calculates the production of clean and dirty inputs, in terms of productivity.
    
\begin{lemma}
    \begin{equation} \label{Ynt}
        Y_{ct}=\left(\frac{\alpha^2}{\psi}\right)^\frac{\alpha}{1-\alpha}\left(A_{ct}^\varphi+A_{dt}^\varphi\right)^{-\frac{\alpha+\varphi}{\varphi}}A_{ct}A_{dt}^{\alpha+\varphi}
    \end{equation}
    \begin{equation} \label{Yot}
        Y_{dt}=\left(\frac{\alpha^2}{\psi}\right)^\frac{\alpha}{1-\alpha}\left(A_{ct}^\varphi+A_{dt}^\varphi\right)^{-\frac{\alpha+\varphi}{\varphi}}A_{ct}^{\alpha+\varphi}A_{dt}
    \end{equation}
\end{lemma}

\begin{proof}
    Plugging (\ref{xjit_2}) into (\ref{Yjt}) gives:
    \begin{multline} \label{Yjt_2}
        Y_{jt}=L_{jt}^{1-\alpha}\int_0^1A_{jit}^{1-\alpha}x_{jit}^\alpha \; di=\left(\frac{\alpha^2}{\psi}\right)^\frac{\alpha}{1-\alpha}L_{jt}^{1-\alpha}\int_0^1A_{jit}^{1-\alpha}p_{jt}^\frac{\alpha}{1-\alpha}L_{jt}^\alpha A_{jit}^\alpha \; di\\
        =\left(\frac{\alpha^2}{\psi}\right)^\frac{\alpha}{1-\alpha}p_{jt}^\frac{\alpha}{1-\alpha}L_{jt}\int_0^1A_{jit} \; di=\left(\frac{\alpha^2}{\psi}\right)^\frac{\alpha}{1-\alpha}p_{jt}^\frac{\alpha}{1-\alpha}L_{jt}A_{jt}
    \end{multline}
    
    Then, plugging (\ref{pnt}) and (\ref{Lnt}) into (\ref{Yjt_2}):
    \begin{multline*}
        Y_{ct}=\left(\frac{\alpha^2}{\psi}\right)^\frac{\alpha}{1-\alpha} \left(\frac{A_{ct}^{-\varphi}}{A_{ct}^{-\varphi}+A_{dt}^{-\varphi}}\right)^{\frac{\alpha}{\varphi}} \frac{A_{ct}^{-\varphi}}{A_{ct}^{-\varphi}+A_{dt}^{-\varphi}} A_{ct} =\left(\frac{\alpha^2}{\psi}\right)^\frac{\alpha}{1-\alpha} \left(\frac{A_{ct}^\varphi A_{dt}^\varphi}{A_{ct}^\varphi A_{dt}^\varphi} \times \frac{A_{ct}^{-\varphi}}{A_{ct}^{-\varphi}+A_{dt}^{-\varphi}}\right)^{\frac{\alpha+\varphi}{\varphi}} A_{ct} \\
        = \left(\frac{\alpha^2}{\psi}\right)^\frac{\alpha}{1-\alpha}\left(\frac{A_{dt}^{\varphi}}{A_{ct}^{\varphi}+A_{dt}^{\varphi}}\right)^{\frac{\alpha+\varphi}{\varphi}} A_{ct} =  \left(\frac{\alpha^2}{\psi}\right)^\frac{\alpha}{1-\alpha}\left(A_{ct}^\varphi+A_{dt}^\varphi\right)^{-\frac{\alpha+\varphi}{\varphi}}A_{ct}A_{dt}^{\alpha+\varphi}
    \end{multline*}
    
    thus giving (\ref{Ynt}). A similar line of argument gives (\ref{Yot}).
\end{proof}

    We then calculate the production of the final good.
    
\begin{lemma}
    \begin{equation} \label{Yt_2}
        Y_t=\left(\frac{\alpha^2}{\psi}\right)^\frac{\alpha}{1-\alpha}\left(A_{ct}^\varphi+A_{dt}^\varphi\right)^{-\frac{1}{\varphi}}A_{ct}A_{dt}
    \end{equation}
\end{lemma}

\begin{proof}
    Using (\ref{Yt}), (\ref{Ynt}) and (\ref{Yot}):
    \begin{multline*}
        Y_t=\left(\frac{\alpha^2}{\psi}\right)^\frac{\alpha}{1-\alpha}\left\{\left[\left(A_{ct}^\varphi+A_{dt}^\varphi\right)^{-\frac{\alpha+\varphi}{\varphi}}A_{ct}A_{dt}^{\alpha+\varphi}\right]^\frac{\sigma-1}{\sigma}+\left[\left(A_{ct}^\varphi+A_{dt}^\varphi\right)^{-\frac{\alpha+\varphi}{\varphi}}A_{ct}A_{dt}^{\alpha+\varphi}\right]^\frac{\sigma-1}{\sigma}\right\}^\frac{\sigma}{\sigma-1}\\
        = \left(\frac{\alpha^2}{\psi}\right)^\frac{\alpha}{1-\alpha}\left(A_{ct}^\varphi+A_{dt}^\varphi\right)^{-\frac{\alpha+\varphi}{\varphi}}A_{ct}A_{dt}\left(A_{ct}^{(\alpha+\varphi-1)\frac{\sigma-1}{\sigma}}+A_{dt}^{(\alpha+\varphi-1)\frac{\sigma-1}{\sigma}}\right)^\frac{\sigma}{\sigma-1}\\
        = \left(\frac{\alpha^2}{\psi}\right)^\frac{\alpha}{1-\alpha}\left(A_{ct}^\varphi+A_{dt}^\varphi\right)^{-\frac{\alpha+\varphi}{\varphi}}A_{ct}A_{dt}\left(A_{ct}^\varphi+A_{dt}^\varphi\right)^\frac{\sigma}{\sigma-1} =\left(\frac{\alpha^2}{\psi}\right)^\frac{\alpha}{1-\alpha} \left(A_{ct}^\varphi+A_{dt}^\varphi\right)^{\frac{\sigma}{\sigma-1}-\frac{\alpha+\varphi}{\varphi}}A_{ct}A_{dt}\\
        = \left(\frac{\alpha^2}{\psi}\right)^\frac{\alpha}{1-\alpha}\left(A_{ct}^\varphi+A_{dt}^\varphi\right)^{-\frac{1}{\varphi}}A_{ct}A_{dt}
    \end{multline*}
\end{proof}

    Next, we calculate the total consumption by consumers.
    
\begin{lemma}
    \begin{equation}
        C_t=\left(\frac{\alpha^2}{\psi}\right)^\frac{\alpha}{1-\alpha}(1-\alpha^2)\left(A_{ct}^\varphi+A_{dt}^\varphi\right)^{-\frac{1}{\varphi}}A_{ct}A_{dt}
    \end{equation}
\end{lemma}

\begin{proof}
    Using (\ref{nt}), (\ref{xjit_2}) and (\ref{Yt_2}):
    \begin{multline*}
        C_t=\left(\frac{\alpha^2}{\psi}\right)^\frac{\alpha}{1-\alpha}\left(A_{ct}^\varphi+A_{dt}^\varphi\right)^{-\frac{1}{\varphi}}A_{ct}A_{dt} - \psi\left(\frac{\alpha^2}{\psi}\right)^\frac{1}{1-\alpha}p_{ct}^\frac{1}{1-\alpha}L_{ct}A_{ct} - \psi\left(\frac{\alpha^2}{\psi}\right)^\frac{1}{1-\alpha}p_{dt}^\frac{1}{1-\alpha}L_{dt}A_{dt}\\
        = \left(\frac{\alpha^2}{\psi}\right)^\frac{\alpha}{1-\alpha}\Biggl[\left(A_{ct}^\varphi+A_{dt}^\varphi\right)^{-\frac{1}{\varphi}}A_{ct}A_{dt} - \alpha^2\left(\frac{A_{ct}^{-\varphi}}{A_{ct}^{-\varphi}+A_{dt}^{-\varphi}}\right)^\frac{\varphi+1}{\varphi}A_{ct} - \alpha^2\left(\frac{A_{dt}^{-\varphi}}{A_{ct}^{-\varphi}+A_{dt}^{-\varphi}}\right)^\frac{\varphi+1}{\varphi}A_{dt}\Biggl]\\
        = \left(\frac{\alpha^2}{\psi}\right)^\frac{\alpha}{1-\alpha}\Biggl[\left(A_{ct}^\varphi+A_{dt}^\varphi\right)^{-\frac{1}{\varphi}}A_{ct}A_{dt} - \alpha^2\left(\frac{A_{dt}^{\varphi}}{A_{ct}^{\varphi}+A_{dt}^{\varphi}}\right)^\frac{\varphi+1}{\varphi}A_{ct} - \alpha^2\left(\frac{A_{ct}^{\varphi}}{A_{ct}^{\varphi}+A_{dt}^{\varphi}}\right)^\frac{\varphi+1}{\varphi}A_{dt}\Biggl]\\
        = \left(\frac{\alpha^2}{\psi}\right)^\frac{\alpha}{1-\alpha}\Biggl[\left(A_{ct}^\varphi+A_{dt}^\varphi\right)^{-\frac{1}{\varphi}}A_{ct}A_{dt} - \alpha^2\left(A_{ct}^\varphi+A_{dt}^\varphi\right)^{-1-\frac{1}{\varphi}}A_{dt}^{\varphi+1}A_{ct} - \alpha^2\left(A_{ct}^\varphi+A_{dt}^\varphi\right)^{-1-\frac{1}{\varphi}}A_{ct}^{\varphi+1}A_{dt}\Biggl]\\
        = \left(\frac{\alpha^2}{\psi}\right)^\frac{\alpha}{1-\alpha}\left(A_{ct}^\varphi+A_{dt}^\varphi\right)^{-\frac{1}{\varphi}}A_{ct}A_{dt}\Biggl[1 - \alpha^2\left(A_{ct}^\varphi+A_{dt}^\varphi\right)^{-1}A_{dt}^{\varphi} - \alpha^2\left(A_{ct}^\varphi+A_{dt}^\varphi\right)^{-1}A_{ct}^{\varphi}\Biggl]\\
        = \left(\frac{\alpha^2}{\psi}\right)^\frac{\alpha}{1-\alpha}(1-\alpha^2)\left(A_{ct}^\varphi+A_{dt}^\varphi\right)^{-\frac{1}{\varphi}}A_{ct}A_{dt}
    \end{multline*}
\end{proof}

\subsection{Innovation in Equilibrium}

Finally, we used the lemmas above to calculate the ratio between expected profits of scientists in the clean and dirty sectors.
    
\begin{lemma}
    \begin{equation} \label{pijt}
        \frac{\Pi_{ct}}{\Pi_{dt}}=\frac{\eta_c}{\eta_d}\times \frac{1+\gamma I_{ct}}{1+\gamma I_{dt}}\times \left(\frac{1+\gamma \eta_c I_{ct} s_{ct}}{1+\gamma \eta_d I_{dt} s_{dt}}\right)^{-\varphi-1}\times \left(\frac{A_{ct-1}}{A_{dt-1}}\right)^{-\varphi}
    \end{equation}
\end{lemma}

\begin{proof}
    From (\ref{pijt_3}), we have:
    \begin{equation}
        \frac{\Pi_{ct}}{\Pi_{dt}}=\frac{\eta_c}{\eta_d}\times \frac{1+\gamma I_{ct}}{1+\gamma I_{dt}} \times  \left(\frac{p_{ct}}{p_{dt}}\right)^\frac{1}{1-\alpha}\times \frac{L_{ct}}{L_{dt}}\times \frac{A_{ct-1}}{A_{dt-1}}
    \end{equation}
    
    Plugging (\ref{pjt_2}) and (\ref{Ljt}) into the above:
    \begin{multline*}
        \frac{\Pi_{ct}}{\Pi_{dt}}=\frac{\eta_c}{\eta_d}\times \frac{1+\gamma I_{ct}}{1+\gamma I_{dt}}\times \left(\frac{A_{ct}}{A_{dt}}\right)^{-1}\times \left(\frac{A_{ct}}{A_{dt}}\right)^{-\varphi}\times \frac{A_{ct-1}}{A_{dt-1}}=\frac{\eta_c}{\eta_d}\times \frac{1+\gamma I_{ct}}{1+\gamma I_{dt}}\times \left(\frac{A_{ct}}{A_{dt}}\right)^{-\varphi-1}\times \frac{A_{ct-1}}{A_{dt-1}}\\
        =\frac{\eta_c}{\eta_d}\times \frac{1+\gamma I_{ct}}{1+\gamma I_{dt}}\times \left(\frac{1+\gamma \eta_c I_{ct} s_{ct}}{1+\gamma \eta_d I_{dt} s_{dt}}\right)^{-\varphi-1}\times \left(\frac{A_{ct-1}}{A_{dt-1}}\right)^{-\varphi-1}\times \frac{A_{ct-1}}{A_{dt-1}}\\
        =\frac{\eta_c}{\eta_d}\times \frac{1+\gamma I_{ct}}{1+\gamma I_{dt}}\times \left(\frac{1+\gamma \eta_c I_{ct} s_{ct}}{1+\gamma \eta_d I_{dt} s_{dt}}\right)^{-\varphi-1}\times \left(\frac{A_{ct-1}}{A_{dt-1}}\right)^{-\varphi}
    \end{multline*}
\end{proof}

In equilibrium, scientists conduct research in the sector with higher expected profit. The above lemma thus gives rise to the following important theorem on the allocation of scientists.
    
\begin{theorem}
    At equilibrium, innovation happens in:
    \begin{itemize}
        \item only the clean sector if and only if
        \begin{equation} \label{clean}
            \frac{\eta_c}{\eta_d}\times \frac{1+\gamma I_{ct}}{1+\gamma I_{dt}}\times \left(1+\gamma \eta_c I_{ct}\right)^{-\varphi-1}\times \left(\frac{A_{ct-1}}{A_{dt-1}}\right)^{-\varphi}\geq 1
        \end{equation}
        \item only the dirty sector if and only if
        \begin{equation} \label{dirty}
            \frac{\eta_c}{\eta_d}\times \frac{1+\gamma I_{ct}}{1+\gamma I_{dt}}\times \left(1+\gamma \eta_d I_{dt} \right)^{\varphi+1}\times \left(\frac{A_{ct-1}}{A_{dt-1}}\right)^{-\varphi}\leq 1
        \end{equation}
        \item both sectors if and only if
        \begin{equation} \label{both}
            \frac{\eta_c}{\eta_d}\times \frac{1+\gamma I_{ct}}{1+\gamma I_{dt}}\times \left(\frac{1+\gamma \eta_c I_{ct} s_{ct}}{1+\gamma \eta_d I_{dt} s_{dt}}\right)^{-\varphi-1}\times \left(\frac{A_{ct-1}}{A_{dt-1}}\right)^{-\varphi}=1
        \end{equation}
        for some $s_{ct}+s_{dt}=1$
    \end{itemize}
\end{theorem}

\begin{proof}
    We recall that in the free market, scientists will research in the sector that brings more expected profit. Now, after fixing $t$, we define:
    \begin{equation}
        f(s)=\frac{\eta_c}{\eta_d}\times \frac{1+\gamma I_{ct}}{1+\gamma I_{dt}}\times \left(\frac{1+\gamma \eta_c I_{ct} s}{1+\gamma \eta_c I_{dt} (1-s)}\right)^{-\varphi-1}\times \left(\frac{A_{ct-1}}{A_{dt-1}}\right)^{-\varphi}
    \end{equation}
    
    Then we can rewrite (\ref{pijt}) as $\Pi_{ct}/\Pi_{dt}=f(s_{ct})$.
    
    Firstly, $f(1)\geq 1$ is equivalent to $s=1$ being the equilibrium, which means innovation occurs in only the clean sector if $f(1)\geq 1$. This yields (\ref{clean}).
    
    Next, $f(0)\leq 1$ is equivalent to $s=0$ being the equilibrium, meaning innovation only occurs in the dirty sector. This yields (\ref{dirty}).
    
    Now assume we have $f(1)< 1 < f(0)$. Note that $f(s)$ is continuous, and depending on the value of $1+\varphi$, $f(s)$ is strictly increasing, strictly decreasing, or is a constant. The inequality implies that $f(s)$ is strictly decreasing, which means there exists a unique $s'\in(0,1)$ so that $f(s')=1$. This means innovation happens in both sectors if $f(s_{ct})=1$ for some $s_{ct}\in [0,1]$. Conversely, if innovation occurs in both sector, then there must exist $s_{ct}+s_{dt}=1$ so that $f(s_{ct})=1$. Hence we proved both directions, yielding (\ref{both}).
\end{proof}

\section{The Impact on The Environment}

We will show and give an intuition of why AI alone, without government intervention, may not be able to avert an environmental disaster. Only when AI is used in tandem with a sufficiently large but temporary government intervention will an environmental disaster be avoided. 

\subsection{Environmental Disaster Might Not be Averted}
Without AI, Acemoglu et al. (2012) shows that with clean technology being sufficiently backwards compared to dirty technology and $\sigma>1$, an environmental disaster necessarily occurs under economic equilibrium without government intervention. We will show that with AI, the situation is not so different:

\begin{theorem}
    If $s_{cT}=0$ for any $T$ such that:
    $$T \geq \frac{1}{k}\ln\frac{e^{-\frac{4k}{\varphi}}-1}{\gamma \eta_d}-1$$
    then $s_{ct}=0$ for all $t>T$.
\end{theorem}

\begin{proof}
We will use induction: Consider a sufficiently large time $t$ where $s_{ct}=0$. According to (\ref{dirty}), this means that:
\begin{equation*}
    \frac{\eta_c}{\eta_d}\times \frac{1+\gamma e^{3kt}}{1+\gamma e^{kt}}\times \left(1+\gamma \eta_d e^{kt} \right)^{\varphi+1}\times \left(\frac{A_{ct-1}}{A_{dt-1}}\right)^{-\varphi}\leq 1
\end{equation*}

Let the left hand side be $R_t$. Consider $R_{t+1}$, taking note that $s_{ct}=0$ and $s_{dt}=1$:
\begin{multline*}
    R_{t+1} = \frac{\eta_c}{\eta_d}\times \frac{1+\gamma e^{3k(t+1)}}{1+\gamma e^{k(t+1)}}\times \left(1+\gamma \eta_d e^{k(t+1)} \right)^{\varphi+1}\times \left(\frac{A_{ct}}{A_{dt}}\right)^{-\varphi} \\
    = \frac{\eta_c}{\eta_d}\times \frac{1+\gamma e^{3k(t+1)}}{1+\gamma e^{k(t+1)}}\times \left(1+\gamma \eta_d e^{k(t+1)} \right)^{\varphi+1}\times \left(1+\gamma \eta_d e^{kt}\right)^{\varphi} \times \left(\frac{A_{ct-1}}{A_{dt-1}}\right)^{-\varphi}
\end{multline*}

We then have the ratio:
\begin{equation} \label{rt}
    \frac{R_{t+1}}{R_t} = \frac{\left(1+\gamma e^{3k(t+1)}\right) \left(1+\gamma e^{kt}\right)}{\left(1+\gamma e^{k(t+1)}\right) \left(1+\gamma e^{3kt}\right)} \times \frac{\left(1+\gamma \eta_d e^{k(t+1)} \right)^{\varphi+1}}{1+\gamma \eta_d e^{kt}}
\end{equation}

At this point, we prove the following 2 lemmas:

\begin{lemma}\label{e^3k}
    $$\frac{\left(1+\gamma e^{3k(t+1)}\right) \left(1+\gamma e^{kt}\right)}{\left(1+\gamma e^{k(t+1)}\right) \left(1+\gamma e^{3kt}\right)} \leq e^{3k}$$
\end{lemma}

\begin{proof}
    The expression to be proven is equivalent with:
    $$\Leftrightarrow\ (1+\gamma e^{3k(t+1)})(1+\gamma e^{kt}) \leq (1+\gamma e^{k(t+1)})(1+\gamma e^{3kt})(e^{3k})$$
    $$\Leftrightarrow\ 1+\gamma e^{3kt+3k} + \gamma e^{kt} + \gamma^2 e^{4kt+3k} \leq e^{3k} + \gamma e^{3kt+3k} + \gamma e^{kt+4k} + \gamma^2 e^{4kt+7k}$$
    which is trivially true.
\end{proof}

\begin{lemma}\label{1/e^3k}
    If:
    $$t\geq \frac{1}{k}\ln\frac{e^{-\frac{4k}{\varphi}}-1}{\gamma \eta_d}-1$$
    then:
    $$\frac{\left(1+\gamma \eta_d e^{k(t+1)} \right)^{\varphi+1}}{1+\gamma \eta_d e^{kt}} \leq e^{-3k}$$
\end{lemma}

\begin{proof}
    Noting that $\varphi<0$, the given condition of $t$ implies:
    $$\Rightarrow e^{k(t+1)} \geq \frac{e^{-\frac{4k}{\varphi}}-1}{\gamma \eta_d}$$
    $$\Rightarrow 1 + \gamma \eta_d e^{k(t+1)} \geq e^{-\frac{4k}{\varphi}}$$
    $$\Rightarrow \left(1 + \gamma \eta_d e^{k(t+1)}\right)^\varphi \leq e^{-4k}$$
    $$\Rightarrow \left(1 + \gamma \eta_d e^{k(t+1)}\right)^{\varphi+1} \leq e^{-4k}\left(1 + \gamma \eta_d e^{k(t+1)}\right)$$
    $$\Rightarrow \left(1 + \gamma \eta_d e^{k(t+1)}\right)^{\varphi+1} \leq e^{-4k}\left(e^{kt} + \gamma \eta_d e^{k(t+1)}\right)$$
    $$\Rightarrow \left(1 + \gamma \eta_d e^{k(t+1)}\right)^{\varphi+1} \leq e^{-3k}\left(1 + \gamma \eta_d e^{kt}\right)$$
    and we are done.
\end{proof}

We return to our main Theorem. Using Lemma \ref{e^3k} and \ref{1/e^3k}, we conclude that $R_{t+1}\leq R_t$ for all $t$ as given in the Theorem. Thus equation (\ref{dirty}) still applies for $t+1$, and $s_{ct+1}=0$. Induction completes our proof.

\end{proof}

This result differs slightly from Acemoglu et al. (2012) in the sense that it is still possible for AI development to change the trend of innovation from the dirty to the clean sector. However, any such change either happens sufficiently early, or does not happen at all.

This result unfortunately leads to a pessimistic prediction:

\begin{theorem}
    If $\sigma>1$ and $s_{cT}=0$ for any $T$ such that:
    $$T \geq \frac{1}{k}\ln\frac{e^{-\frac{4k}{\varphi}}-1}{\gamma \eta_d}-1$$
    then $S_t=0$ for all sufficiently large $t$.
\end{theorem}

\begin{proof}
    The above Theorem shows that research occurs in the dirty sector only for sufficiently large $t$. This means that for all sufficiently large $t$:
    \begin{equation*}
        A_{dt+1} = (1+ \gamma \eta_d) A_{dt}
    \end{equation*}

    which means that $A_{dt}$ grows at a rate of $\gamma \eta_d$. In other words:
    \begin{equation*}
        \frac{A'_{dt}}{A_{dt}} = \gamma \eta_d
    \end{equation*}
    
    Now, taking ln of both sides of (\ref{Yot}):
    $$\ln Y_{dt}=-\frac{\alpha+\varphi}{\varphi}\ln(A^\varphi_{ct}+A^\varphi_{dt})+(\alpha+\varphi)\ln A_{ct}+\ln A_{dt} + const$$
    
    Partially differentiate w.r.t $A_{ct}$, we get the following for sufficiently large $t$:
    $$\frac{Y'_{dt}}{Y_{dt}}=-\frac{\alpha+\varphi}{\varphi}\times \frac{\varphi A^{\varphi-1}_{dt}A'_{dt}}{A^\varphi_{ct}+A^\varphi_{dt}}+\frac{A'_{dt}}{A_{dt}} = \frac{A'_{dt}}{A_{dt}}\left( 1-(\alpha+\varphi)\times \frac{A^\varphi_{dt}}{A^\varphi_{ct}+A^\varphi_{dt}}\right)=\gamma \eta_d \left( 1-(\alpha+\varphi)\times \frac{A^\varphi_{dt}}{A^\varphi_{ct}+A^\varphi_{dt}}\right)$$
    
    In the long run, $A_{dt}$ tends to infinity. Since $\varphi < 0$, $A_{dt}^\varphi$ tends to 0. With $A_{ct}^\varphi$ being a constant, $A_{dt}^\varphi/(A_{ct}^\varphi+A_{dt}^\varphi)$ tends to 0.

    Plugging that in the above equation, we get $Y'_{dt}/Y_{dt}=\gamma \eta_d$ in the long run. This means that the growth rate of $Y_{dt}$ is $\gamma \eta_d$, implying that $Y_{dt}$ grows to infinity.

    As a result, $Y_{dt}>(1+\rho)\xi^{-1}\bar{S}$ for all sufficiently large $t$. Equation (\ref{St}) then gives:
    \begin{equation*}
        S_{t+1}=-\xi Y_{dt}+(1+\rho)S_t<-(1+\rho)\bar{S} + (1+\rho)\bar{S} = 0
    \end{equation*}
    
    for all sufficiently large $t$, resulting in an environmental disaster.
\end{proof}

\subsection{Why Environmental Disaster Avoidance is Uncertain}

We provide an intuition as to why AI may not be able to reverse an environmental disaster. There are two forces in play in our model.

The first one is AI having the potential to benefit clean research 3 times more than dirty research. This means that for a successful research, the percentage increase in machine quality is much larger in the clean sector. As a result, the expected profit of a researcher in the clean sector significantly increases compared to in the dirty sector. This might be able to incentivise researchers to switch to the clean sector, which increases productivity and drives clean production.

Mathematically, this impact of AI is reflected by the term $(1+\gamma I_{jt})$ in equation (\ref{pijt_3}), which is carried to the term $(1+\gamma I_{ct})/(1+\gamma I_{dt})$ in (\ref{pijt}). In Theorem 2, when considering how $\Pi_{ct}/\Pi_{dt}$ changes over time, this impact is considered in the first term of (\ref{rt}).

However, there is also a second force in play, which is that AI likely only benefit dirty production in reality. Acemoglu et al. (2012) assumes that dirty production is sufficiently ahead of clean production so that all research initially happens in the dirty sector. Thus, until that changes, only dirty production actually benefits from AI.

Considering the allocation of scientists, this aspect is seen in the term $A_{ct-1}/A_{dt-1}$ in (\ref{pijt}). Only $A_{dt-1}$ would increase, and since this term has a positive exponent, this factor makes dirty research more profitable. In Theorem 2, this impact is reflected in the second term of (\ref{rt}).

We note that even though clean production may not benefit from AI at all from the start, this does not mean that research cannot switch to the clean sector. If the potential impact of AI on clean production (represented by $(1+\gamma I_{ct})$) grows faster than the actual impact of AI on dirty production (represented by $A_{dt}$), the ratio $\Pi_{ct}/\Pi_{dt}$ can still increase and become larger than 1, tipping research to the clean sector.

However, this scenario may not happen. For a sufficiently large $t$, both the first and second terms of (\ref{rt}) are bounded by constants. This means that for a sufficiently large $t$, there is an upper bound to how fast AI's potential to benefit clean research can increase, an a lower bound to the growth rate of productivity in the dirty sector. In such a scenario, research cannot switch to the clean sector. We also note that in the scenario where the switch do happen, it must happen early at a small $t$, or it won't happen at all 

\subsection{Environmental Disaster Is Averted With Government Intervention}

We now prove that AI, with the help of government intervention, allows us to avoid an environmental disaster. We first prove the following theorem, which is essentially a mirror of Theorem 2.

\begin{theorem}
    If $s_{cT}=1$ for
    $$T \geq \max \left(\frac{1}{2k}\ln \frac{\gamma (e^{2k} + e^k + 1) + 2}{\gamma e^{2k}}, \frac{1}{3k}\ln \frac{e^{-\frac{k}{\varphi}}-1}{\gamma \eta_c}-1 \right)$$
    then $s_{ct}=1$ for all $t>T$.
\end{theorem}

\begin{proof}
The proof is structurally identical to that of Theorem 2. We will once again use induction: Consider a sufficiently large time $t$ where $s_{ct}=1$. According to (\ref{clean}), this means that:
\begin{equation*}
    \frac{\eta_c}{\eta_d}\times \frac{1+\gamma e^{3kt}}{1+\gamma e^{kt}}\times \left(1+\gamma \eta_c e^{3kt} \right)^{-\varphi-1}\times \left(\frac{A_{ct-1}}{A_{dt-1}}\right)^{-\varphi}\geq 1
\end{equation*}

Let the left hand side be $R_t$. Consider $R_{t+1}$, taking note that $s_{ct}=1$ and $s_{dt}=0$:
\begin{multline*}
    R_{t+1} = \frac{\eta_c}{\eta_d}\times \frac{1+\gamma e^{3k(t+1)}}{1+\gamma e^{k(t+1)}}\times \left(1+\gamma \eta_c e^{3k(t+1)} \right)^{-\varphi-1}\times \left(\frac{A_{ct}}{A_{dt}}\right)^{-\varphi} \\
    = \frac{\eta_c}{\eta_d}\times \frac{1+\gamma e^{3k(t+1)}}{1+\gamma e^{k(t+1)}}\times \left(1+\gamma \eta_c e^{3k(t+1)} \right)^{-\varphi-1}\times \left(1+\gamma \eta_c e^{3kt}\right)^{-\varphi} \times \left(\frac{A_{ct-1}}{A_{dt-1}}\right)^{-\varphi}
\end{multline*}

We then have the ratio:
\begin{equation}
    \frac{R_{t+1}}{R_t} = \frac{\left(1+\gamma e^{3k(t+1)}\right) \left(1+\gamma e^{kt}\right)}{\left(1+\gamma e^{k(t+1)}\right) \left(1+\gamma e^{3kt}\right)} \times \frac{1+\gamma \eta_c e^{3kt}}{\left(1+\gamma \eta_c e^{3k(t+1)} \right)^{\varphi+1}}
\end{equation}

Now we prove the following Lemmas:

\begin{lemma} \label{e^2k}
    If:
    $$t \geq \frac{1}{2k}\ln \frac{\gamma (e^{2k} + e^k + 1) + 2}{\gamma e^{2k}}$$
    then:
    $$\frac{\left(1+\gamma e^{3k(t+1)}\right) \left(1+\gamma e^{kt}\right)}{\left(1+\gamma e^{k(t+1)}\right) \left(1+\gamma e^{3kt}\right)} \geq e^{2k}$$
\end{lemma}

\begin{proof}
    Manipulating the given condition of $t$, we get:
    $$\gamma e^{2kt+2k} \geq \gamma (e^{2k} + e^k + 1) + 2$$

    Now, note that since $e^k \geq 1$, we have $1 \leq e^k \leq e^{kt} \Rightarrow 2 \geq (e^k+1)/e^{kt}$. Hence:
    $$\gamma e^{2kt+2k} \geq \gamma (e^{2k} + e^k + 1) + \frac{e^k+1}{e^{kt}}$$
    $$\gamma e^{3kt+2k} \geq \gamma (e^{2k} + e^k + 1) e^{kt} + e^k+1$$

    This expression is actually equivalent to what we want to prove:
    $$\gamma (e^k-1) e^{3kt+2k} \geq \gamma (e^k-1)(e^{2k} + e^k + 1) e^{kt} + (e^k-1)(e^k+1)$$
    $$\gamma e^{3kt+3k} - \gamma e^{3kt+2k} \geq \gamma (e^{3k}-1) e^{kt} + e^{2k} - 1$$
    $$\gamma e^{3kt+3k} + \gamma e^{kt} + 1 \geq \gamma e^{3kt+2k} + \gamma e^{kt+3k} + e^{2k}$$

    which becomes clear after we add a term to factorise each side:
    $$\gamma^2 e^{4kt+3k} + \gamma e^{3kt+3k} + \gamma e^{kt} + 1 \geq \gamma^2 e^{4kt+3k} + \gamma e^{3kt+2k} + \gamma e^{kt+3k} + e^{2k}$$
    $$\left(1+\gamma e^{3k(t+1)}\right) \left(1+\gamma e^{kt}\right) \geq e^{2k} \left(1+\gamma e^{k(t+1)}\right) \left(1+\gamma e^{3kt}\right)$$

    This completes our proof.
\end{proof}

\begin{lemma}\label{1/e^2k}
    If:
    $$t \geq \frac{1}{3k}\ln \frac{e^{-\frac{k}{\varphi}}-1}{\gamma \eta_c}-1$$
    then:
    $$\frac{1+\gamma \eta_c e^{3kt}}{\left(1+\gamma \eta_c e^{3k(t+1)} \right)^{\varphi+1}} \geq e^{-2k}$$
\end{lemma}

\begin{proof}
    Noting that $\varphi<0$, the given condition of $t$ implies:
    $$\Rightarrow 3k(t+1) \geq \ln \frac{e^{-\frac{k}{\varphi}}-1}{\gamma \eta_c}$$
    $$\Rightarrow e^{3k(t+1)} \geq \frac{e^{-\frac{k}{\varphi}}-1}{\gamma \eta_c}$$
    $$\Rightarrow 1 + \gamma \eta_c e^{3k(t+1)} \geq e^{-\frac{k}{\varphi}}$$
    $$\Rightarrow \left(1 + \gamma \eta_c e^{3k(t+1)}\right)^\varphi \leq e^{-k}$$
    $$\Rightarrow \left(1 + \gamma \eta_c e^{3k(t+1)}\right)^{\varphi+1} \leq e^{-k}\left(1 + \gamma \eta_c e^{3k(t+1)}\right)$$
    $$\Rightarrow \left(1 + \gamma \eta_c e^{3k(t+1)}\right)^{\varphi+1} \leq e^{-k}\left(e^{3k} + \gamma \eta_c e^{3k(t+1)}\right)$$
    $$\Rightarrow \left(1 + \gamma \eta_c e^{3k(t+1)}\right)^{\varphi+1} \leq e^{2k}\left(1 + \gamma \eta_c e^{3kt}\right)$$
    and we are done.
\end{proof}


We return to our main Theorem. Using Lemma \ref{e^2k} and \ref{1/e^2k}, we conclude that $R_{t+1}\geq R_t$ for all $t$ as given in the Theorem. Thus equation (\ref{clean}) still applies for $t+1$, and $s_{ct+1}=1$. Induction completes our proof.
\end{proof}

The above theorem means that for a sufficiently far point in the future, if all scientists are working in the clean sector, then it remains the case at all times from that point onwards.

This gives rise to the following theorem on the effect of government intervention:

\begin{theorem}
    If $\bar{S}$ is sufficiently high, then a sufficiently large but only temporary subsidy for the clean sector and/or taxation on the dirty sector will prevent an environmental disaster.
\end{theorem}

\begin{proof}
    A subsidy for the clean sector can be seen as a percentage increase in expected profit of scientists in the clean sector. Similarly, a taxation on the dirty sector is interpreted as a percentage decrease in expected profit of scientists in the dirty sector.

    Let the combined effect of the subsidy and/or taxation at time $t$ be expressed by $G_t$ so that:
    \begin{equation}
        \frac{\Pi_{ct}}{\Pi_{dt}}=G_t \times \frac{\eta_c}{\eta_d}\times \frac{1+\gamma I_{ct}}{1+\gamma I_{dt}}\times \left(\frac{1+\gamma \eta_c I_{ct} s_{ct}}{1+\gamma \eta_d I_{dt} s_{dt}}\right)^{-\varphi-1}\times \left(\frac{A_{ct-1}}{A_{dt-1}}\right)^{-\varphi}
    \end{equation}

    Assume that $T$ is the sufficiently large constant in Theorem 4. Then for all $i=1,...,T$, we set $G_t$ so that $\Pi_{ct}/\Pi_{dt}\geq 1$. This demonstrates that the government intervention is only temporary.

    By theorem 4, we deduce that $s_{ct}=1$ for all $t$. This means that $s_{dt}=0$, and thus $A'_{ct}=0$.

    The rest of the proof closely mirrors that of Theorem 3. Taking ln of both sides of (\ref{Yot}):
    $$\ln Y_{dt}=-\frac{\alpha+\varphi}{\varphi}\ln(A^\varphi_{ct}+A^\varphi_{dt})+(\alpha+\varphi)\ln A_{ct}+\ln A_{dt}$$
    
    Partially differentiate w.r.t $A_{ct}$:
    $$\frac{Y'_{dt}}{Y_{dt}}=-\frac{\alpha+\varphi}{\varphi}\times \frac{\varphi A^{\varphi-1}_{dt}A'_{dt}}{A^\varphi_{ct}+A^\varphi_{dt}}+\frac{A'_{dt}}{A_{dt}} = \frac{A'_{dt}}{A_{dt}}\left( 1-(\alpha+\varphi)\times \frac{A^\varphi_{dt}}{A^\varphi_{ct}+A^\varphi_{dt}}\right)=0$$
    
    Thus $Y'_{dt}=0$, meaning $Y_{dt}$ is unchanged for all $t$. Equation (\ref{St}) then immediately implies that for a sufficiently large $\bar{S}$, $S_t>0$ for all $t$.
\end{proof}

In essence, the taxation on the dirty sector and/or the subsidy on the clean sector provide a financial incentive for scientists to move from the dirty to the clean sector. This allows innovation to happen in the clean sector instead.

Moreover, recall that there exists a tipping point where if innovation only happen in the clean sector at that point, then innovation will only happen in the clean sector in the future. Thus, government intervention only needs to continue until this tipping point, and is therefore temporary.

It should be noted that this is a more promising result than in Acemoglu et al. (2012) as Proposition 3 in this paper shows that temporary subsidies can only prevent an environmental disaster if clean and dirty input are strong substitutes.

\section{Computer Simulation}

\subsection{Choices of Parameters}

Our parameter choices are mostly similar to Acemoglu et al. (2012), with some minor changes. For simulation constants, we set the period length to 2 and the number of periods to 50 to investigate the next 100 years.

Regarding constants on the economy, $\sigma = 10$ and $\alpha = 1/3$. We normalise $\psi = \alpha^2 = 1/9$ without loss of generality.

Regarding constants on the innovation process, we set $\gamma=1$ and $\eta_c=\eta_d=0.02$ so that productivity without AI increases by 2\% per year.

The value of $Y_{c0}$ and $Y_{d0}$ are taken as `the production of nonfossil and fossil fuel in the world primary energy supply from 2002 to 2006' in quadrillion of British thermal unit (Btu). We then calculate the corresponding values of $A_{c0}$ and $A_{d0}$ as follows: from equations (\ref{Ynt}) and (\ref{Yot}):
\begin{equation}
    \frac{Y_{ct}}{Y_{dt}} = \left(\frac{A_{dt}}{A_{ct}}\right)^{\alpha+\varphi-1} \Rightarrow  A_{dt}= \left(\frac{Y_{ct}}{Y_{dt}}\right)^\frac{1}{\alpha+\varphi-1} A_{ct} = rA_{ct}
\end{equation}

Therefore, equation (\ref{Ynt}) becomes:
\begin{equation*}
    Y_{ct} = \left(\frac{\alpha^2}{\psi}\right)^\frac{\alpha}{1-\alpha} (A_{ct}^\varphi + (rA_{ct})^\varphi)^{-\frac{\alpha+\varphi}{\varphi}} A_{ct} (rA_{ct})^{\alpha+\varphi} = \left(\frac{\alpha^2}{\psi}\right)^\frac{\alpha}{1-\alpha} (1 + r^{-\varphi})^{-\frac{\alpha+\varphi}{\varphi}} A_{ct}
\end{equation*}
\begin{equation}
    \Rightarrow A_{ct}= \left(\frac{\alpha^2}{\psi}\right)^{-\frac{\alpha}{1-\alpha}} (1 + r^{-\varphi})^{\frac{\alpha+\varphi}{\varphi}} Y_{ct}
\end{equation}

and then we have $A_{dt}=rA_{ct}$.

For $k$, we will calculate 4 of its possible values, based on 4 different prediction on how AI will affect GDP. Here, in the interest of simplicity, we assume that consumption $C_t$ is GDP. We summarise the predictions below:

\begin{table}[h]
    \centering
    \begin{tabular}{|c|c|}
        \hline
        \textbf{Authors} & \textbf{AI contribution to GDP in next 10 years} \\
        \hline
        Acemoglu (2024) & 0.9\% \\
        \hline
        Goldman Sachs (2023) & 7\% \\
        \hline
        McKinsey \& Company (2023) & 33.2\% \\
        \hline
        Korinek \& Suh (2024) & 64.4\% \\
        \hline
    \end{tabular}
    \caption{Predictions of AI contribution to GDP}
\end{table}

For environment-related parameters, we first define:
\begin{equation}
    \Delta_t = 3\log_2 \frac{CO2_t}{280}
\end{equation}

where at time $t$, $\Delta_t$ is the increase in temperature compared to the pre-industrial period, and $CO2_t$ is the $CO_2$ concentration in the atmosphere.

We define the environmental disaster as an increase of $6^oC$ in temperature, meaning $\Delta_{disaster}=6$, giving $CO2_{disaster}=1120$. 

The environmental quality is then defined as:
\begin{equation}
    S_t=CO2_{disaster}-\max(CO2_t,280)
\end{equation}

Moreover, we define $S_0=\bar{S}-99$ instead, to reflect that the current atmospheric $CO_2$ concentration has increased by 99 ppm.

Finally, we calculate $\xi$ as the $CO_2$ emission per unit of dirty input produced from 2002 to 2006, and $\delta$ so that only half of the $CO_2$ emitted from production contributes to the increase in atmospheric $CO_2$ concentration.

\subsection{Results}

The values of $k$ are as follows. We also calculate the lower bound for $T$ in Theorem 4, which represents the minimum number of years the government needs to intervene to switch research entirely to the clean sector.

\begin{table}[H]
    \centering
    \begin{tabular}{|c|c|c|c|c|c|}
        \hline
        \textbf{No.} & \textbf{Authors} & \textbf{AI impact} & \textbf{$k$} & \textbf{Intervention} & \textbf{Avoid disaster?}\\
        \hline
        1 & Acemoglu (2024) & 0.9\% & 0.0160 & 100 years & No \\
        \hline
        2 & Goldman Sachs (2023) & 7\% & 0.102 & 15 years & Yes \\
        \hline
        3 & McKinsey \& Company (2023) & 33.2\% & 0.180 & 8 years & Yes \\
        \hline
        4 & Korinek \& Suh (2024) & 64.4\% & 0.201 & 7 years & Yes \\
        \hline
    \end{tabular}
    \caption{Simulation Results}
\end{table}

Only the first, most pessimistic prediction results in an environmental disaster; all others show that we will avoid an environmental disaster.

Figure 1 shows that the first scenario will see all research being conducted in the dirty sector, as expected. For other scenarios, research completely switch to the clean sector, but at different times. In the second scenario, it takes around 20 years, while for the third and fourth one, this change takes place almost immediately.

\begin{figure}[H]
    \centering
    \includegraphics[width=0.5\linewidth]{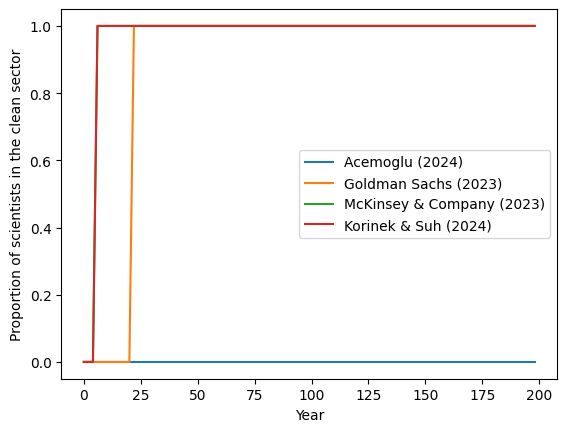}
    \caption{Proportion of scientists in the clean sector}
\end{figure}

Naturally, this has a decisive impact on production. Figure 2 and 3 shows that when researches are all in the dirty sector (as in scenario 1), production of dirty input skyrockets, and clean input contribution to final good production is virtually zero. In the opposite scenario, when researches are all in the clean sector, production of dirty input is virtually zero, and almost all of the final good is produced from clean inputs.

\begin{figure}[H]
    \centering
    \includegraphics[width=0.5\linewidth]{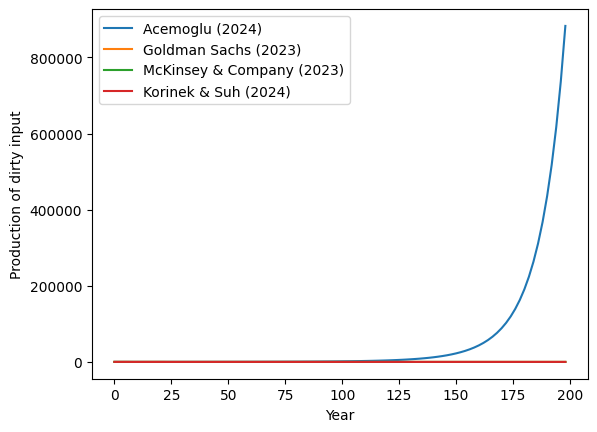}
    \caption{Production of dirty input}
\end{figure}

\begin{figure}[H]
    \centering
    \includegraphics[width=0.5\linewidth]{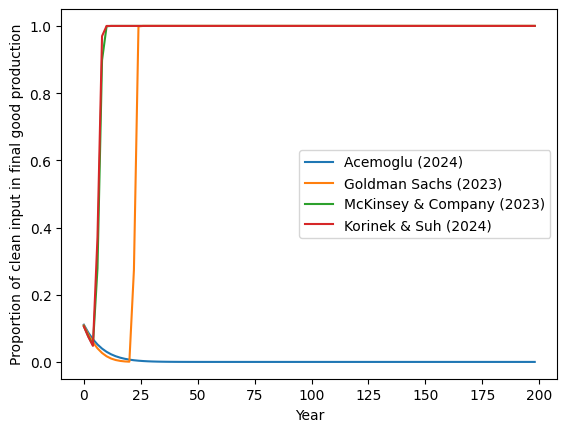}
    \caption{Proportion of clean input in final good production}
\end{figure}

Figure 4 demonstrates that when researches and production switch to the clean sector, the rate at which consumption increases is significantly larger (note that we cap consumption at $10^{100}$ to avoid overflow issues during simulation). This is because researches and production are better able to utilize AI, when AI has a larger impact on the clean sector than the dirty sector.

\begin{figure}[H]
    \centering
    \includegraphics[width=0.5\linewidth]{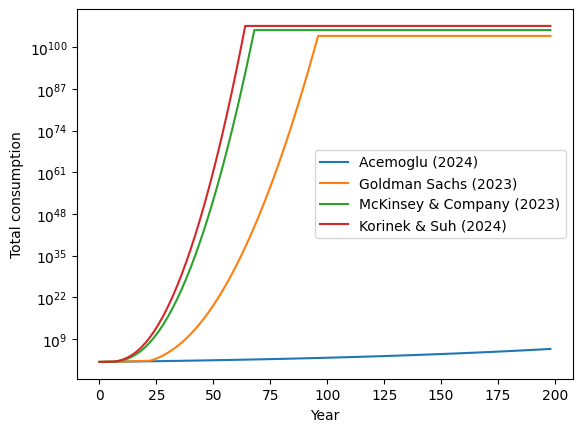}
    \caption{Total consumption}
\end{figure}

Figure 5 and 6 shows the effect on the environment in those scenario. With Acemoglu's prediction, the temperature will gradually increase starting from the $90^{th}$ year (it only stops at 6 degrees Celcius because we cap the temperature increase at that point). As a result, an environmental disaster will occur in around 125 years. It should be noted that during the first 90 years, dirty production is still increasing ``behind the scene", so that by the $90^{th}$, its detrimental environmental impact outweigh the regenerate ability of the environment.

For other scenarios, our simulation gives almost identical results, with temperature returning to that of pre-industrial period and an environmental disaster is avoided.

\begin{figure}[H]
    \centering
    \includegraphics[width=0.5\linewidth]{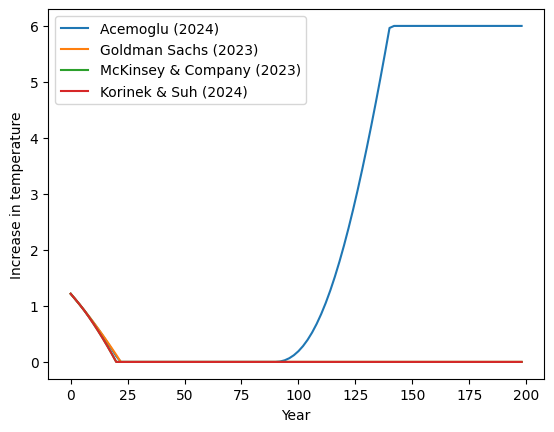}
    \caption{Increase in temperature}
    \label{fig:enter-label}
\end{figure}

\begin{figure}[H]
    \centering
    \includegraphics[width=0.5\linewidth]{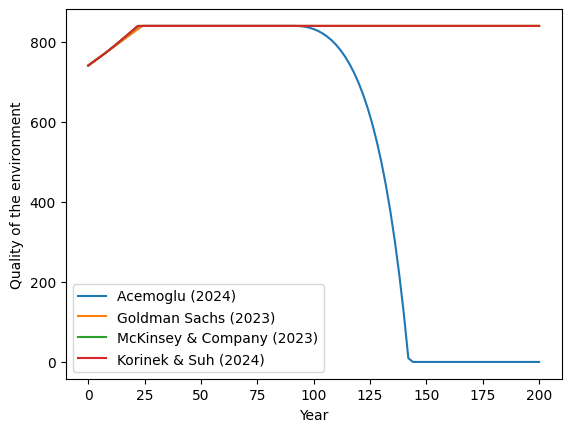}
    \caption{Quality of the environment}
\end{figure}

\section{Conclusion}
With climate change being one of the most serious and urgent topic nowadays, there have been many attempts at new, innovative ways to protect the environment. Many of such projects leverage the huge potential of AI, leading some to believe that AI will lead our environmental efforts. However, many are more cautious and doubt the extent to which AI can contribute to the economy and the environment.

This paper partially validates those doubt, as it shows quantitatively that AI might not be able to avert an environmental disaster, and argue qualitatively that it is unlikely to do so. Only when paired with temporary government intervention can AI definitively prevent an environmental disaster.

That said, this paper also shows quantitatively that AI allows for more optimistic prediction. Acemoglu et al. (2012) shows that an environmental disaster will definitely occur without government intervention, while with AI, there is a chance that it will push clean production just enough to prevent such a disaster. Moreover, Acemoglu et al. proves that a temporary subsidy can only prevent an environmental disaster if clean and dirty input are strong substitutes. With AI, a temporary subsidy (or taxation on the dirty sector) is guaranteed to achieve that goal.

This paper is one of the earliest attempt at quantifying the impact of AI on the environment. Future studies can continue to improve upon this paper by considering different growth rate of AI (e.g. logistic growth), different ways to quantify the impact of AI in the model, and prove stronger mathematical results that highlight the potential of AI in protecting the environment.

\section*{Code}
The code used for computer simulation is available at the following Github repository:

\texttt{https://github.com/ld-minh4354/Impact-of-AI-on-Environmental-Quality}

\section*{Declarations}

\textbf{Funding:} The authors did not receive support from any organization for the submitted work.

\noindent
\textbf{Conflict of Interest:} The authors have no relevant financial or non-financial interests to disclose.

\noindent
\textbf{Ethical Approval:} This article does not contain any studies with human participants or animals performed by any of the authors.


\begin{thebibliography}{9}

\bibitem{}
Acemoglu, D., Aghion, P., Bursztyn, L., \& Hemous, D. (2012). The environment and directed technical change. \textit{American Economic Review}, 102(1), 1
31-166. https://doi.org/10.1257/aer.102.1.131

\bibitem{}
Acemoglu, D. (2024). The Simple Macroeconomics of AI (No. w32487). National Bureau of Economic Research.

\bibitem{}
Alderman, K., Turner, L. R., \& Tong, S. (2012). Floods and human health: a systematic review. Environment international, 47, 37-47.

\bibitem{}
Andres, P., Dugoua, E., \& Dumas, M. (2022). Directed technological change and general purpose technologies: can AI accelerate clean energy innovation?.

\bibitem{}
Batten, S. (2018). Climate change and the macro-economy: a critical review.

\bibitem{}
Bergmann, D. (2023). What is self-supervised learning?. IBM. https://www.ibm.com/topics/self-supervised-learning

\bibitem{}
Bruzzone, A., \& Orsoni, A. (2003, March). AI and simulation-based techniques for the assessment of supply chain logistic performance. In 36th Annual Simulation Symposium, 2003. (pp. 154-164). IEEE.

\bibitem{}
Cabrera, Á. A., Perer, A., \& Hong, J. I. (2023). Improving human-AI collaboration with descriptions of AI behavior. Proceedings of the ACM on Human-Computer Interaction, 7(CSCW1), 1-21. 

\bibitem{}
Chantry, M., Christensen, H., Dueben, P., \& Palmer, T. (2021). Opportunities and challenges for machine learning in weather and climate modelling: hard, medium and soft AI. Philosophical Transactions of the Royal Society A, 379(2194), 20200083.

\bibitem{}
Choi, J. D., \& Kim, K. E. (2011). Inverse reinforcement learning in partially observable environments. Journal of Machine Learning Research, 12, 691-730.

\bibitem{}
Chui, M., Hazan, E., Roberts, R., Singla, A., Smaje, K., Sukharevsky, A., Yee, L., \& Zemmel, R. (2023). The economic potential of Generative AI: The Next Productivity Frontier. McKinsey \& Company. https://www.mckinsey.com/capabilities/mckinsey-digital/our-insights/the-economic-potential-of-generative-ai-the-next-productivity-frontier 

\bibitem{}
Clarke, B., Otto, F., Stuart-Smith, R., \& Harrington, L. (2022). Extreme weather impacts of climate change: an attribution perspective. Environmental Research: Climate, 1(1), 012001.

\bibitem{}
Crimmins, A., Balbus, J., Gamble, J. L., Beard, C. B., Bell, J. E., Dodgen, D., ... \& Ziska, L. (2016). The impacts of climate change on human health in the United States: a scientific assessment. The Impacts of Climate Change on Human Health in the United States: A Scientific Assessment.

\bibitem{}
Falkowski, P., Scholes, R. J., Boyle, E. E. A., Canadell, J., Canfield, D., Elser, J., ... \& Steffen, W. (2000). The global carbon cycle: a test of our knowledge of earth as a system. \textit{Science}, 290(5490), 291-296. https://doi.org/10.1126/science.290.5490.291

\bibitem{}
Farzaneh, H., Malehmirchegini, L., Bejan, A., Afolabi, T., Mulumba, A., \& Daka, P. P. (2021). Artificial intelligence evolution in smart buildings for energy efficiency. Applied Sciences, 11(2), 763.

\bibitem{}
Fok, S. C., \& Ong, E. K. (1996). A high school project on artificial intelligence in robotics. \textit{Artificial Intelligence in Engineering}, 10(1), 61-70. https://doi.org/10.1016/0954-1810(95)00016-X

\bibitem{}
Folini, D., Kübler, F., Malova, A., \& Scheidegger, S. (2021). The climate in climate economics. arXiv preprint arXiv:2107.06162. https://doi.org/10.2139/ssrn.3885021

\bibitem{}
Galea, S., Brewin, C. R., Gruber, M., Jones, R. T., King, D. W., King, L. A., ... \& Kessler, R. C. (2007). Exposure to hurricane-related stressors and mental illness after Hurricane Katrina. Archives of general psychiatry, 64(12), 1427-1434.

\bibitem{}
Goldman Sachs (2023). Generative AI could raise global GDP by 7\%. https://www.goldmansachs.com/intelligence/pages/generative-ai-could-raise-global-gdp-by-7-percent.html 

\bibitem{}
Hagendorff, T., \& Wezel, K. (2020). 15 challenges for AI: or what AI (currently) can’t do. Ai \& Society, 35(2), 355-365.

\bibitem{}
Hansen, G., \& Stone, D. (2016). Assessing the observed impact of anthropogenic climate change. Nature Climate Change, 6(5), 532-537.

\bibitem{}
Hemming, S., Zwart, F. D., Elings, A., Petropoulou, A., \& Righini, I. (2020). Cherry tomato production in intelligent greenhouses—Sensors and AI for control of climate, irrigation, crop yield, and quality. Sensors, 20(22), 6430.

\bibitem{}
Hritonenko, N., \& Yatsenko, Y. (2009). Mathematical models of global trends and technological change. MATHEMATICAL MODELS–Volume III, 2, 303

\bibitem{}
Jaffe, A. B., Newell, R. G., \& Stavins, R. N. (1999). Energy-efficient technologies and climate change policies: issues and evidence. Available at SSRN 198829. https://doi.org/10.2139/ssrn.198829

\bibitem{}
Joksimovic, S., Ifenthaler, D., Marrone, R., De Laat, M., \& Siemens, G. (2023). Opportunities of artificial intelligence for supporting complex problem-solving: Findings from a scoping review. Computers and Education: Artificial Intelligence, 4, 100138.

\bibitem{}
Kallis, G. (2011). In defence of degrowth. Ecological economics, 70(5), 873-880.

\bibitem{}
Kaufmann, R. K., Kauppi, H., Mann, M. L., \& Stock, J. H. (2011). Reconciling anthropogenic climate change with observed temperature 1998–2008. Proceedings of the National Academy of Sciences, 108(29), 11790-11793.

\bibitem{}
Klein, P., \& Bergmann, R. (2019, July). Generation of Complex Data for AI-based Predictive Maintenance Research with a Physical Factory Model. In ICINCO (1) (pp. 40-50).

\bibitem{}
Korinek, A., \& Suh, D. (2024). Scenarios for the Transition to AGI (No. w32255). National Bureau of Economic Research.

\bibitem{}
Kotlikoff, L., Kubler, F., Polbin, A., \& Scheidegger, S. (2021a). Pareto-improving carbon-risk taxation. \textit{Economic Policy}, 36(107), 551-589. https://doi.org/10.1093/epolic/eiab008

\bibitem{}
Kotlikoff, L., Kubler, F., Polbin, A., \& Scheidegger, S. (2021b). Making carbon taxation a global win-win. \textit{No Brainers and Low-Hanging Fruit in National Climate Policy}. Retrieved March 15, 2024, from https://cepr.org/system/files/publication-files/110107-no\_brainers\_and\_low\_hanging\_fruit\_in\_national\_climate\_policy.pdf\#page=234

\bibitem{}
Kotlikoff, L. J., Kubler, F., Polbin, A., \& Scheidegger, S. (2021c). Can today's and tomorrow's world uniformly gain from carbon taxation? (No. w29224). \textit{National Bureau of Economic Research}. https://doi.org/10.3386/w29224

\bibitem{}
Kwon, C., Raman, A., \& Moreno, A. (2023). The Impact of Input Inaccuracy on Leveraging AI Tools: Evidence from Algorithmic Labor Scheduling. Available at SSRN.

\bibitem{}
Lawrance, E., Thompson, R., Fontana, G., \& Jennings, N. (2021). The impact of climate change on mental health and emotional wellbeing: current evidence and implications for policy and practice. Grantham Institute briefing paper, 36, 1-36.

\bibitem{}
Liu, G., Catacutan, D. B., Rathod, K., Swanson, K., Jin, W., Mohammed, J. C., ... \& Stokes, J. M. (2023). Deep learning-guided discovery of an antibiotic targeting Acinetobacter baumannii. Nature Chemical Biology, 19(11), 1342-1350.

\bibitem{}
Lu, S., Bai, X., Zhang, X., Li, W., \& Tang, Y. (2019). The impact of climate change on the sustainable development of regional economy. Journal of Cleaner Production, 233, 1387-1395.

\bibitem{}
Malek, M. A. (2022). Criminal courts’ artificial intelligence: the way it reinforces bias and discrimination. AI and Ethics, 2(1), 233-245.

\bibitem{}
Manning, C., \& Clayton, S. (2018). Threats to mental health and wellbeing associated with climate change. In Psychology and climate change (pp. 217-244). Academic Press.

\bibitem{}
Morgan, J. (2020). Degrowth: necessary, urgent and good for you. Real-World Economics Review, (93), 113-131.

\bibitem{}
Muškardin, E., Tappler, M., Aichernig, B. K., \& Pill, I. (2023, November). Reinforcement learning under partial observability guided by learned environment models. In International Conference on Integrated Formal Methods (pp. 257-276). Cham: Springer Nature Switzerland.

\bibitem{}
NASA. (2023). Global Temperature | Vital Signs – Climate Change: Vital Signs of the Planet. Climate Change. Retrieved July 12, 2024, from https://climate.nasa.gov/vital-signs/global-temperature/?intent=121

\bibitem{}
Nordhaus, W. (2018a). Projections and uncertainties about climate change in an era of minimal climate policies. \textit{American Economic Journal: Economic Policy}, 10(3), 333-360. https://doi.org/10.1257/pol.20170046

\bibitem{}
Nordhaus, W. (2018b). Evolution of modeling of the economics of global warming: changes in the DICE model, 1992–2017. \textit{Climatic Change}, 148(4), 623-640. https://doi.org/10.1007/s10584-018-2218-y

\bibitem{}
Patel, R. K., Kumari, A., Tanwar, S., Hong, W. C., \& Sharma, R. (2022). AI-empowered recommender system for renewable energy harvesting in smart grid system. IEEE Access, 10, 24316-24326.

\bibitem{}
Popp, D. (2002). Induced innovation and energy prices. \textit{American Economic Review}, 92(1), 160-180. https://doi.org/10.1257/000282802760015658

\bibitem{}
Rosenzweig, C., Karoly, D., Vicarelli, M., Neofotis, P., Wu, Q., Casassa, G., ... \& Imeson, A. (2008). Attributing physical and biological impacts to anthropogenic climate change. Nature, 453(7193), 353-357.

\bibitem{}
Selvaraj, M. G., Vergara, A., Ruiz, H., Safari, N., Elayabalan, S., Ocimati, W., \& Blomme, G. (2019). AI-powered banana diseases and pest detection. Plant methods, 15, 1-11.


\bibitem{}
Sekulova, F., Kallis, G., Rodríguez-Labajos, B., \& Schneider, F. (2013). Degrowth: from theory to practice. Journal of cleaner Production, 38, 1-6.

\bibitem{}
Sirotkin, K., Carballeira, P., \& Escudero-Viñolo, M. (2022). A study on the distribution of social biases in self-supervised learning visual models. In Proceedings of the IEEE/CVF Conference on Computer Vision and Pattern Recognition (pp. 10442-10451).

\bibitem{}
Sohrabi, S., Udrea, O., \& Riabov, A. (2013, June). Hypothesis exploration for malware detection using planning. In Proceedings of the AAAI Conference on Artificial Intelligence (Vol. 27, No. 1, pp. 883-889).

\bibitem{}
de Souza, D. C., Cabrita, L., Galinha, C. F., Rato, T. J., \& Reis, M. S. (2021). A Spectral AutoML approach for industrial soft sensor development: Validation in an oil refinery plant. Computers \& Chemical Engineering, 150, 107324.

\bibitem{}
Spaan, M. T. J. (2006). Approximate planning under uncertainty in partially observable environments. Universiteit van Amsterdam [Host].

\bibitem{}
Srinivasan, S., Lanctot, M., Zambaldi, V., Pérolat, J., Tuyls, K., Munos, R., \& Bowling, M. (2018). Actor-critic policy optimization in partially observable multiagent environments. Advances in neural information processing systems, 31.

\bibitem{}
Tamburrini, G. (2022). The AI carbon footprint and responsibilities of AI scientists. Philosophies, 7(1), 4.

\bibitem{}
Stern, David I., and Robert K. Kaufmann. "Anthropogenic and natural causes of climate change." Climatic change 122 (2014): 257-269.

\bibitem{}
Tol, R. S. (2018). The economic impacts of climate change. Review of environmental economics and policy.

\bibitem{}
Vössing, M., Kühl, N., Lind, M., \& Satzger, G. (2022). Designing transparency for effective human-AI collaboration. Information Systems Frontiers, 24(3), 877-895.

\bibitem{}
Wang, H., Fu, T., Du, Y., Gao, W., Huang, K., Liu, Z. et al. (2023). Scientific discovery in the age of artificial intelligence. Nature, 620(7972), 47-60.

\bibitem{}
Westphal, M., Vössing, M., Satzger, G., Yom-Tov, G. B., \& Rafaeli, A. (2023). Decision control and explanations in human-AI collaboration: Improving user perceptions and compliance. Computers in Human Behavior, 144, 107714.

\bibitem{}
Wilson, H. J., \& Daugherty, P. R. (2018). Collaborative intelligence: Humans and AI are joining forces. Harvard Business Review, 96(4), 114-123.

\bibitem{}
World Bank. (2024). GDP Data. https://data.worldbank.org/indicator/NY.GDP.MKTP.CD

\bibitem{}
Xiang, X., Li, Q., Khan, S., \& Khalaf, O. I. (2021). Urban water resource management for sustainable environment planning using artificial intelligence techniques. Environmental Impact Assessment Review, 86, 106515.

\bibitem{}
Xu, Y., Liu, X., Cao, X., Huang, C., Liu, E. et al (2021). Artificial intelligence: A powerful paradigm for scientific research. The Innovation, 2(4).

\bibitem{}
Yazdani-Asrami, M., Sadeghi, A., Song, W., Madureira, A., Murta-Pina, J., Morandi, A., \& Parizh, M. (2022). Artificial intelligence methods for applied superconductivity: material, design, manufacturing, testing, operation, and condition monitoring. Superconductor Science and Technology, 35(12), 123001.

\bibitem{}
Zhang, Y., Liao, Q. V., \& Bellamy, R. K. (2020, January). Effect of confidence and explanation on accuracy and trust calibration in AI-assisted decision making. In Proceedings of the 2020 conference on fairness, accountability, and transparency (pp. 295-305).

\bibitem{}
Zhang, D., Maslej, N., Brynjolfsson, E., Etchemendy, J., Lyons, T., Manyika, J., ... \& Perrault, R. (2022). The AI index 2022 annual report. AI index steering committee. \textit{Stanford Institute for Human-Centered AI, Stanford University, 123}. Retreived March 14, 2024, from https://aiindex.stanford.edu/ai-index-report-2022/

\end{thebibliography}
\end{document}